\documentclass{article}
\usepackage{amsmath,amsthm,amsfonts,graphicx}
\usepackage{url,xspace}
\usepackage{rotating}
\newcommand{\dqt}{{DQT}\xspace}
\newcommand{\fqt}{{FQT}\xspace}

\newtheorem{example}{Example}
\newtheorem{theorem}{Theorem}
\newtheorem{lemma}{Lemma}

\usepackage[ruled,nothing]{algorithm}
\usepackage{algorithmic}

\newcommand{\GF}[1][vide]{\ifthenelse{\equal{#1}{vide}}{}{\ensuremath{\mathtt {GF}(#1)}}}
\newcommand{\N}{\ensuremath{\mathbb N}}
\newcommand{\Z}{\ensuremath{\mathbb Z}}

\newcommand{\pF}[1][vide]{\ifthenelse{\equal{#1}{vide}}{\Z}{\leavevmode
	\kern.1em\raise.0ex \hbox{\Z}\kern-.1em /\kern-.15em\lower.3ex
         \hbox{#1\mbox{\Z}}}
}

\title{Q-adic Transform Revisited}

\newcommand{\alignauthor}{}
\newcommand{\affaddr}[1]{{\em #1}}
\newcommand{\email}[1]{\url{#1}}

\author{
\alignauthor Jean-Guillaume Dumas\\
\affaddr{Universit\'e de Grenoble,}
\affaddr{Laboratoire J. Kuntzmann,}
\affaddr{umr CNRS 5224.}\\
\affaddr{BP 53X, 51, rue des Math\'ematiques.}
\affaddr{F38041 Grenoble, France.}\\
\email{Jean-Guillaume.Dumas@imag.fr}
}

\begin{document}
\maketitle
\begin{abstract} We present an algorithm to perform a simultaneous
  modular reduction of several residues. This enables to compress
  polynomials into integers and perform several modular operations with
  machine integer arithmetic. The
  idea is to convert the $X$-adic representation of modular
  polynomials, with $X$ an indeterminate, to a $q$-adic representation
  where $q$ is an integer larger than the field characteristic. 
  With some control on the different involved sizes it
  is then possible to perform some of the $q$-adic arithmetic directly
  with machine integers or floating points. Depending also on the number of
  performed numerical operations one can then convert back to the
  $q$-adic or $X$-adic representation and eventually mod out high
  residues.
  In this note we present a new version of both conversions: more
  tabulations and a way to reduce the number of divisions involved in
  the process are presented. The polynomial multiplication is then
  applied to arithmetic and linear algebra in small finite field extensions.
\end{abstract}



\noindent
\subsection*{\bf Keywords:} Kronecker substitution ; Finite field ; Modular Polynomial
  Multiplication ; REDQ (simultaneous modular reduction) ; Small
  extension field ; DQT (Discrete Q-adic Transform) ; FQT (Fast Q-adic
  Transform).
\section{Introduction}
The FFLAS/FFPACK project has demonstrated the usefulness of
wrapping cache-aware routines for efficient small finite field
linear algebra \cite{jgd:2002:fflas,jgd:2004:ffpack}.
 
A conversion between a modular representation of prime fields
and e.g. floating points used exactly is natural. It uses the homomorphism
to the integers. Now for extension
fields (isomorphic to polynomials over a prime field) such a
conversion is not direct. In \cite{jgd:2002:fflas} we proposed 
transforming the polynomials into a $q$-adic representation where $q$ is
an integer larger than the field characteristic. We call this
transformation \dqt for Discrete Q-adic Transform, it is a form of
Kronecker substitution \cite[\S 8.4]{VonzurGathen:1999:MCA}. 
With some care, in particular on
the size of $q$, it is possible to map the operations in the extension
field into the floating point arithmetic realization of this $q$-adic
representation and convert back using an inverse \dqt.

In this note we propose some implementation
improvements: we propose to use a tabulated discrete
logarithm for the \dqt and we give a trick to reduce the number of machine
divisions involved in the inverse. 
This then gives rise to an improved
\dqt which we thus call \fqt (Fast Q-adic Transform). This \fqt uses a
simultaneous reduction of several residues, called REDQ, and some
table lookup.

Therefore we recall in section~\ref{sec:galois} the previous
conversion algorithm and discuss in section \ref{sec:floor} about
a floating point implementation of modular reduction. This
implementation will be
used throughout the paper to get fast reductions. 
We then present our new simultaneous reduction in section
\ref{sec:redq} and show in section~\ref{sec:tmto} 
how a time-memory trade-off can make this
reduction very fast. This fast reduction is then
applied to modular polynomial
multiplication with small prime fields in section~\ref{sec:del}. It is
also applied to small extension field arithmetic and fast matrix
multiplication over those fields in section~\ref{sec:fflas}.

\section{Q-adic representation of \\ polynomials}\label{sec:galois}
We follow here the presentation of \cite{jgd:2002:fflas} of the idea
of \cite{Saunders:2001:qadic}: polynomial
arithmetic is performed a $q-$adic way, with $q$ a sufficiently big prime or
power of a single prime.

Suppose that $a = \sum_{i=0}^{k-1} \alpha_i X^i$ 
and $b = \sum_{i=0}^{k-1} \beta_i X^i$
are two polynomials in $\pF[p][X] $. One can perform
the polynomial multiplication $ a b $ via $q-$adic numbers. 
Indeed, by setting
$\tilde{a} = \sum_{i=0}^{k-1} \alpha_i q^i$ 
and $\tilde{b} = \sum_{i=0}^{k-1} \beta_i q^i$, the product is
computed in the following manner (we suppose that 
$\alpha_i = \beta_i = 0$ for $i > k-1$):
\begin{equation}\label{eq:mult}
\widetilde{a b} = \sum_{j=0}^{2k-2} \left( \sum_{i=0}^{j} \alpha_i
\beta_{j-i} \right) q^j
\end{equation}
Now if $q$ is large enough, the coefficient of $q^i$ will not exceed $q-1$.
In this case, it is possible to evaluate $a$ and $b$ as machine
numbers (e.g. floating point or machine integers), 
compute the product of these evaluations, and convert back to
polynomials by radix computations (see e.g. \cite[Algorithm
9.14]{VonzurGathen:1999:MCA}). There just remains then to perform
modulo $p$ reductions on every coefficient as shown on example
\ref{ex:dqt}.

\begin{example}\label{ex:dqt}
For instance, to multiply $a=X+1$ by $b=X+2$ in $\pF[3][X]$ one can use
the substitution $X=100$: compute $101 \times 102 = 10302$, use radix
conversion to write $10302=q^2+3q+2$ and reduce modulo $3$ to get $a
\times b = X^2+2$.
\end{example}
We call \dqt the evaluation of polynomials modulo $p$ at
$q$ and \dqt inverse the radix conversion of a $q$-adic development
followed by a modular reduction, as shown in algorithm~\ref{alg:dqt}.
\newcounter{horner}
\newcounter{radix}
\newcounter{modin}

\begin{algorithm}[ht]
\caption{Polynomial multiplication by \dqt}\label{alg:dqt}
\begin{algorithmic}[1]
\REQUIRE Two polynomials $v_1$ and $v_2$ in $\pF[p][X] $ of degree
less than $k$.
\REQUIRE a sufficiently large integer $q$.
\ENSURE $R \in \pF[p][X] $, with $R = v_1. v_2$.
\vspace{5pt}\newline{\underline{Polynomial to $q-$adic conversion}}\vspace{2pt}
\STATE Set $\widetilde{v_1}$ and $\widetilde{v_2}$ to the floating point vectors
of the evaluations at $q$ of the elements of $v_1$ and $v_2$.
\setcounter{horner}{\value{ALC@line}}
\COMMENT{Using e.g. Horner's formula}
\vspace{5pt}\newline{\underline{One computation}}\vspace{2pt}
\STATE Compute $\tilde{r} = \widetilde{v_1} \widetilde{v_2}$
\vspace{5pt}\newline{\underline{Building the solution}}\vspace{2pt}
\STATE $\tilde{r} = \sum_{i=0}^{2k-2} \widetilde{\mu_i} q^i$.
\setcounter{radix}{\value{ALC@line}}
\COMMENT{Using radix conversion, see e.g. \cite[Algorithm 9.14]{VonzurGathen:1999:MCA}}
\STATE For each $i$, set $\mu_i = \widetilde{\mu_i} \mod p$
\STATE set $R = \sum_{i=0}^{2k-2} \mu_i X^i $
\setcounter{modin}{\value{ALC@line}}
\end{algorithmic}
\end{algorithm}
Depending on the size of $q$, the results can still
remain exact and we obtain the following bounds generalizing that of
\cite[\S 8.4]{VonzurGathen:1999:MCA}:
\begin{theorem}{\cite{jgd:2002:fflas}} 
Let $m$ be the number of available mantissa bits 
within the machine numbers and $n_q$ be the number of
polynomial products $v_1.v_2$ of degree $k$
accumulated before the re-conversion.
If 
\begin{equation}\label{eq:bounds}
q > n_q k (p-1)^2 \text{~and~} (2k-1)\log_2(q) < m, 
\end{equation} then
Algorithm~\ref{alg:dqt} is correct.
\end{theorem}

Note that the integer $q$ can be chosen to be a power of 2. Then the
  Horner like evaluation (line \thehorner~of
  algorithm~\ref{alg:dqt}) of the polynomials at $q$ 
  is just a left shift. One can then compute
  this shift with exponent manipulations in floating point arithmetic
  and use native shift operator (e.g. the $<<$ operator in C) 
  as soon as values are within
  the $32$ (or $64$ when available) bit range. 

It is shown on \cite[Figures 5 \& 6]{jgd:2002:fflas}
that this wrapping is already a pretty good way to obtain high speed linear
algebra over some small extension fields. Indeed we were able to
reach  high peak performance, quite close to those obtained with
prime fields, namely 420 Millions of finite operations per second
(Mop/s) on a Pentium III, 735 MHz, and more than 500
Mop/s on a 64-bit DEC alpha 500 MHz. This was roughly 20 percent below
the pure floating point performance and 15 percent below the prime
field implementation. 

\section{Euclidean division by floating point routines}\label{sec:floor}
In the implementations of the proposed subsequent algorithms, we will
make extensive use of Euclidean division in exact arithmetic.
Unfortunately exact division is usually quite slow on modern
computers. 
This
division can thus be performed by floating point operations. Suppose
we want to compute $r/p$ where $r$ and $p$ and integers.
Then their difference is representable by a floating
point and, therefore, if $r/p$ is computed by a floating point
division with a rounding to nearest mode, \cite[Theorem
1]{Lefevre:2005:div} assures that
flooring the result gives the expected value. Now if a multiplication
by a precomputed inverse of $p$ is used (as is done e.g. in
NTL \cite{Shoup:2007:NTL}), proving the correctness for all $r$ is more difficult, see
\cite{Lefevre:2005:mul} for more details. 
We therefore propose the following simple lemma which enables the use
of the rounding upward mode to the cost of loosing only one bit of precision:
\begin{lemma}
For two positive integers $p$ and $r$ and $\epsilon>0$, we have
$$\left\lfloor \frac{r}{p} \right\rfloor = \left\lfloor 
\left( r\left(\frac{1}{p}(1+\epsilon)\right)\right)(1+\epsilon) \right\rfloor
~~\text{as long as}~~ r < \frac{1}{2\epsilon+\epsilon^2}.$$
\end{lemma}
\begin{proof} Consider $up\leq r < up+i$ with $u, i$ positive integers and
  $i<p$.
Then $\left\lfloor \frac{r}{p} \right\rfloor =u$ and
$\frac{r}{p}(1+\epsilon)(1+\epsilon)=
u+\frac{i}{p}+\frac{r}{p}(2\epsilon+\epsilon^2)$.
The latter is maximal at $i=p-1$.
This proves that flooring is correct as long as 
$\frac{r}{p}(2\epsilon+\epsilon^2)<\frac{1}{p}$.
\end{proof}
This proves that when rounding towards $+\infty$
it is possible to perform the division by a multiplication
by the precomputed inverse of the prime number as long as 
$r$ is not too large.
Since our entries will be integers but stored 
in floating point format this is a potential significant
speed-up. 

\section{REDQ: modular reduction in the \dqt domain}\label{sec:redq}
The first improvement we propose to the \dqt is to replace the costly
modular reduction of the polynomial coefficients by a {\em single} division
by $p$ (or, better, by a multiplication by its inverse) followed by several
shifts. 
In order to prove the correctness of this algorithm, we first need the
following lemma:
\begin{lemma}\label{lem:flfl} For $r  \in \N$ and $a$, $b \in \N^*$, 
$$\left\lfloor \frac{\left\lfloor
      \frac{r}{b}\right\rfloor}{a}\right\rfloor =
\left\lfloor \frac{r}{ab}\right\rfloor =
\left\lfloor \frac{\left\lfloor \frac{r}{a}\right\rfloor}{b}\right\rfloor$$
\end{lemma}
\begin{proof} 
We proceed by splitting the possible values of $r$ into intervals
$kab \leq r <  (k+1)ab$, where 
$k=\left\lfloor \frac{r}{ab}\right\rfloor$. 
Then $kb \leq \frac{r}{a} < (k+1)b$ and since
$kb$ is an integer we also have that $kb \leq \left\lfloor
  \frac{r}{a}\right\rfloor < (k+1)b$. Thus $k \leq  \frac{\left\lfloor
    \frac{r}{a}\right\rfloor}{b} < k+1$ and $\left\lfloor
  \frac{\left\lfloor \frac{r}{a}\right\rfloor}{b}\right\rfloor
=k$. Obviously the same is true for the left hand side which proves
the lemma.
\end{proof}

This idea is used in algorithm~\ref{alg:REDQ} to perform several
remainderings with a single machine division (note that
when $q$ is a power of $2$, and when elements are represented using an
integral type, division by $q^i$ and flooring are a single
operation, a right shift).

\begin{algorithm}[ht]
\caption{REDQ}\label{alg:REDQ}
\begin{algorithmic}[1]
\REQUIRE Two integers $p$ and $q$ satisfying the conditions (\ref{eq:bounds}).
\REQUIRE $\tilde{r}  = \sum_{i=0}^d \widetilde{\mu_i} q^i \in \Z$.
\ENSURE $\rho \in \Z$, with $\rho = \sum_{i=0}^d \mu_i q^i$
where $\mu_i = \widetilde{\mu_i} \mod p$.
\STATE\label{line:rop} $rop = \left\lfloor \frac{\tilde{r}}{p} \right\rfloor$;
\FOR{$i=0$ to $d$}
\STATE\label{line:ui} $u_i = \left\lfloor \frac{\tilde{r}}{q^i} \right\rfloor -  p \left\lfloor \frac{rop}{q^i} \right\rfloor$;
\ENDFOR
\STATE\label{line:qdbeg} $\mu_{d}=u_{d}$ 
\FOR{$i=0$ to $d-1$}
\STATE $\mu_i = u_i-qu_{i+1} \mod p$;
\ENDFOR\label{line:qdend}
\STATE Return $\rho = \sum_{i=0}^d \mu_i q^i$;
\end{algorithmic}
\end{algorithm}

\begin{theorem}\label{thm:REDQ}
Algorithm REDQ is correct.
\end{theorem}
\begin{proof}
First we need to prove that $0 \leq u_i < p$.
By definition of the truncation, we have $\frac{\tilde{r}}{q^i}-1 < \left\lfloor \frac{\tilde{r}}{q^i}
\right\rfloor \leq \frac{\tilde{r}}{q^i}$ and $\frac{\tilde{r}}{pq^i} - 1 - \frac{1}{q^i}< \left\lfloor \frac{rop}{q^i}
\right\rfloor \leq \frac{\tilde{r}}{pq^i}$.
Thus $-1 < u_i < p +\frac{p}{q^i}$, which is $0 \leq u_i \leq p$ since
$u_i$ is an integer. We now consider the possible case $u_i = p$ and
show that it does not happen. $u_i = p$ means that $\left\lfloor \frac{\tilde{r}}{q^i}
\right\rfloor = p(1+\left\lfloor \frac{rop}{q^i} \right\rfloor) =
pg$. This means that $pgq^i\leq r < pgq^i + q^i$.
So that in turns $gq^i\leq rop \leq \frac{\tilde{r}}{p} <
gq^i+\frac{q^i}{p}$. Thus $g\leq \frac{rop}{q^i}
< g +\frac{1}{p}$ so that $\left\lfloor \frac{rop}{q^i} \right\rfloor
=g$. But then from the definition of $g$ we have that $g = g-1$ which
is absurd. Therefore $0 \leq u_i \leq p-1$.

Second we show that $u_i = \sum_{j=i}^{d} \mu_j q^{j-i} \mod p$. Line
\ref{line:ui} of algorithm \ref{alg:REDQ} defines
$u_i = \left\lfloor \frac{\tilde{r}}{q^i} \right\rfloor -  p \left\lfloor
  \frac{\left\lfloor \frac{\tilde{r}}{p} \right\rfloor}{q^i} \right\rfloor$ and thus lemma~\ref{lem:flfl} gives
that $u_i = \left\lfloor \frac{\tilde{r}}{q^i} \right\rfloor -  p \left\lfloor
  \frac{ \left\lfloor \frac{\tilde{r}}{q^i} \right\rfloor}{p}
\right\rfloor$. The latter is $u_i = \left\lfloor
  \frac{\tilde{r}}{q^i} \right\rfloor \mod p$. Now, since $\tilde{r} =
\sum_{j=0}^{d} \widetilde{\mu_j} q^j$, we have that $\left\lfloor
  \frac{\tilde{r}}{q^i} \right\rfloor = \sum_{j=i}^{d}
\widetilde{\mu_j} q^{j-i}$. Therefore, as $\mu_j = \widetilde{\mu_j}
\mod p$, the equality is proven.
\end{proof}

Note that the last steps are not needed when $p$ divides $q$. Indeed
in this case $q\equiv 0 \mod p$. The trick 
works then simply as shown on example \ref{ex:pdviq} below:
\begin{example}\label{ex:pdviq}
Let $a=X^2+2X+3$ and $b=4X^2+5X+6$ unreduced modulo $5$.
Then $\tilde{a}\times\tilde{b}=40013002800270018$, with $q=10000$,
for which
we need to reduce five coefficients modulo $5$. The trick
is that we can recover all the residues at once. 
Line \ref{line:rop}
produces $rop=\lfloor 08002600560054003.6 \rfloor$. It thus contains
all the quotients $0$;$0002$;$0005$;$0005$;$0003$ and one has then
just to multiply by $p$ and subtract to get:
{\small $\widetilde{a \times b} = 40013002800270018 - 2000500050003 \times
5 = 40003000300020003$} so that $a\times b = 4X^4+3X^3+3X^2+2X+3$.
\end{example}

Now we can give a full example to show the last corrections 
required when $p$ does not
divide $q$. The first part of the algorithm, lines \ref{line:rop} to
\ref{line:ui} is unchanged and is used to get small sizes for $\mu_i$.
The second part is then just a small correction modulo $p$ to get
the correct result.
\begin{example}\label{ex:full}
Take the polynomial $R=1234X^3+5678X^2+9123X+4567$, the
prime $p=23$ and use $q=10^6$. 
In this case, the division gives $rop = \lfloor 1234005678009123004567/23
\rfloor$  $= 53652420783005348024$.
Then the multiplication by the prime produces 
$rop \times 23 =  1234005678009123004552$ so that
$u_0 = 4567-4552=15$. We shift to get $1234005678009123$ and
$53652420783005 \times 23  = 1234005678009115$ which gives $u_1=9123-9115=8$. We
shift and multiply twice to get $u_2=18$ and $u_3= \mu_3 = 15$ just
like in example \ref{ex:pdviq}.
Now $-q = -10^6 \mod 23 = 17$ which is non zero and thus
we have to
compute the corrections of lines \ref{line:qdbeg} to \ref{line:qdend}
of algorithm \ref{alg:REDQ}. This can also be formalized as a matrix
vector product: 
\[ \mu = \left[ \begin{array}{cccc} 1 & 17 & 0 &
    0\\0&1 &17&0\\0&0&1&17\\0&0&0&1\end{array}\right] u \mod p
\]
to get the final result, $R= 15 X^3+20 X^2+15 X+13 \mod 23$.
\end{example}

The algorithm is efficient because one can precompute $1/p$, $1/q$,
$1/q^2$ etc. and use multiplication to compute all of the mods. The
computation of each $u_i$ and $\mu_i$ can also be pipelined or
vectorized since they are independent.
As is, the benefit when
compared to direct remaindering by $p$ is that the corrections occur
on smaller integers. Thus the remaindering by $p$ can be faster.
Actually, another major acceleration can be added:
the fact that the $\mu_i$ are much smaller than the initial
$\tilde{\mu_i}$ makes it possible to tabulate the corrections as shown
next. 

\section{Time-Memory trade-off in REDQ}\label{sec:tmto}
\subsection{A Matrix version of the correction}
Indeed, there is a bijection between the $u_i$ and the $\mu_i$.
This can be viewed on the corrections of lines \ref{line:qdbeg} to
\ref{line:qdend} of algorithm \ref{alg:REDQ}: view these corrections
as a matrix-vector multiplication by a matrix $Q_d$ as in example
\ref{ex:full}. Then we have that:
{\small
\[ 
Q_d = \left[ \begin{array}{ccccc} 
1 	& -q 	& 0	& \ldots 	&0\\
0	& \ddots 	& \ddots	& \ddots 	&\vdots \\
\vdots 	& \ddots	& \ddots 	& \ddots 	&0\\
\vdots 	& 	& \ddots	& \ddots	&-q\\
0 	& \ldots	& \ldots	& 0	&1
\end{array}\right] = 
\left[ \begin{array}{ccccc} 
1 	& q 	& q^2	& \ldots 	&q^{d}\\
0	& \ddots 	& \ddots	& \ddots 	&\vdots \\
\vdots 	& \ddots	& \ddots 	& \ddots 	&q^2\\
\vdots 	& 	& \ddots	& \ddots	&q\\
0	& \ldots	& \ldots	& 0	&1
\end{array}\right] ^{-1}
\]}

\subsection{Tabulations of the matrix-vector product and Time-Memory trade-off}
Thus if the multiplication by $Q_d$ 
is fully tabulated, it requires a table of size at
least $p^{d+1}$. But, due to the nature of $Q_d$, we have the
relations of figure~\ref{fig:qdmat}.
\begin{figure}[ht]\hfill
\includegraphics[height=70pt]{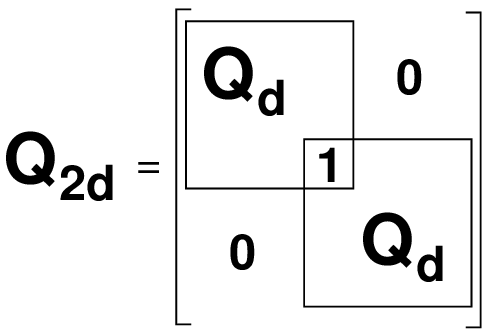}\hfill
\includegraphics[height=70pt]{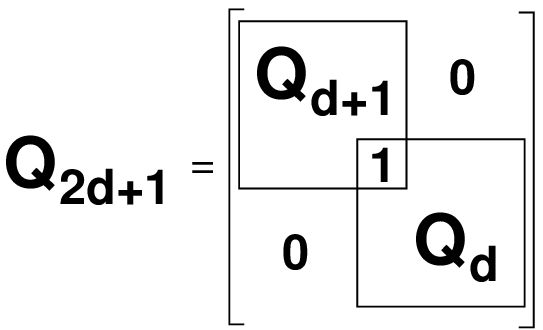}\hfill\ 
\caption{Recurring relations on the $Q_d$ matrices.}\label{fig:qdmat}
\end{figure}

Therefore, it is very easy to tabulate with a table of size $p^k$ only
and perform $\left\lceil 1+\frac{d+1-k}{k-1}\right\rceil =
\left\lceil\frac{d}{k-1}\right\rceil$ 
table accesses
as shown on example \ref{ex:tmto}.
\begin{example}\label{ex:tmto}
Let us compute the corrections for a degree $6$ polynomial. One can
tabulate the multiplication by $Q_6$, a $7\times 7$ matrix, with
therefore $p^7$ entries each of size at least $7\log_2(p)$. Or one can
tabulate the multiplication by $Q_2$, a $3\times 3$ matrix. To compute
$[\mu_0,\ldots,\mu_6]^T = Q_6 [u_0\ldots,u_6]^T$ one can instead use
three multiplications by $Q_2$ and discard the last entry for the
first two multiplications as shown on the
following algorithm:
\begin{algorithm}[H]
\caption{$Q_6$ with an extra memory of size $p^3$}
\begin{algorithmic}[1]
\REQUIRE $[u_0\ldots,u_6] \in \pF[p]^7$.
\REQUIRE The table $Q_2$ of the associated $3\times3$ matrix-vector
multiplication over $\pF[p]$.
\ENSURE $[\mu_0,\ldots,\mu_6]^T = Q_6 [u_0\ldots,u_6]^T$.
\STATE $a_0,a_1,a_2 = Q_2 [ u_0,u_1,u_2]$;
\STATE $b_0,b_1,b_2 = Q_2[u_2,u_3,u_4]$;
\STATE $c_0,c_1,c_2 = Q_2[u_4,u_5,u_6]$;
\STATE Return $[\mu_0,\ldots,\mu_6]^T= [a_0,a_1,b_0,b_1,c_0,c_1,c_2]^T$;
\end{algorithmic}
\end{algorithm}
\end{example}

When $q$ is a power of $2$, 
the computation of the $u_i$, in the first part of algorithm
\ref{alg:REDQ} requires $1~\text{div}$ \& $(d+1)~\text{mul}$ \&
$2d~\text{shifts}$. Now, the time memory trade-off enables to compute
the second part at a choice of costs given on
table~\ref{tab:compl}.
\begin{table}[ht]\center
\begin{tabular}{|l|l|}
\hline
Extra Memory    & time\\
\hline
$0$ & $d$ (mul,add,mod)\\
$p^2$ & $d$ accesses\\
$p^k$ & $\left\lceil \frac{d}{k-1}  \right\rceil$ accesses\\
$p^{d+1}$ & 1 access\\
\hline\end{tabular}
\caption{Time-Memory trade-off in REDQ of degree $d$ over $\pF[p]$}\label{tab:compl}
\end{table}

\subsection{Indexing}
In practice, indexing by a t-uple of integers mod $p$ 
is made by evaluating at $p$, as $\sum u_i p^i$. If more memory
is available, one can also directly index in the binary format
using $\sum u_i \left(2^{\lceil \log_2(p) \rceil}\right)^i$. On the
one hand all
the multiplications by $p$ are replaced by binary shifts. On the other
hand, this makes the table grow a little bit, from $p^k$ to
$2^{\lceil \log_2(p) \rceil k}$. 

\section{Comparison with delayed \\ reduction for polynomial\\ multiplication}\label{sec:del}
The classical alternative to algorithm~\ref{alg:dqt} to perform modular polynomial multiplication is to
use delayed reductions e.g. as in \cite{jgd:2004:dotprod}:
the idea is to accumulate products of the form $\sum_i a_i b_{k-i}$, 
without reductions, while the sum does not overflow. Thus, if we use
for instance a centered representation modulo $p$ (integers from $\frac{1-p}{2}$ to
$\frac{p-1}{2}$), it is possible to accumulate at least $n_d$ products as long as
\begin{equation}\label{eq:delayed}
n_d (p-1)^2 < 2^{m+1}
\end{equation}
The modular reduction can be made by many different ways
(e.g. classical division, floating point multiplication by the
inverse, Montgomery reduction, etc.), we just call the best one REDC
here. It is at most equivalent to 1 machine division.

Now the idea of the \fqt (Fast Q-adic Transform) 
is to represent modular polynomials 
of the form $P = \sum_{i=0}^N a_i X^i$ by 
$P = \sum P_i \left( X^{d+1} \right) ^{i}$ where the $P_i$ are degree $d$
  polynomials stored in a single integer in the $q$-adic way.

Therefore, a product $PQ$ has the form 
$$\sum \left(\sum P_i Q_{t-i}\right) \left(X^{d+1}\right)^{t}.$$
There, each multiplication $P_i Q_{t-i}$ is made by algorithm
\ref{alg:dqt} on a single machine integer.
The reduction is made by a tabulated REDQ and 
can also be delayed now as long as conditions
(\ref{eq:bounds}) are guaranteed.

We want to compare these two strategies. We thus propose the following
complexity model: 
we count only multiplications and additions in the field as an
atomic operation and separate the machine divisions. We for instance
approximate REDC by a machine division. We call REDQ$_k$ a
simultaneous reduction of $k$ residues. In our complexity model a
REDQ$_k$ thus requires  $1$ division and $2k$ multiplications and
additions. 
We also call $d-$FQT the use of a degree
$d$ $q$-adic substitution. Thus a multiplication $P_i Q_{t-i}$ in a
$d-$FQT requires the reduction of $2d+1$ coefficients, i.e. a
REDQ$_{2d+1}$.

Let $P$ be a polynomial of degree $N$ with indeterminate "X". 
If we use a $d-$FQT, it will then become a polynomial of degree $D_q$
in the indeterminate $Y=X^{d+1}$. 
Thus, $$D_q = \left\lceil \frac{N+1}{d+1} \right\rceil -1.$$
Table~\ref{tab:mulpol} gives the respective
complexities of both strategies. 

\begin{table}[ht]\center
\begin{tabular}{|l||c|c|}
\hline
& Mul \& Add & Reductions \\
\hline
Delayed & $(2N+1)^2$ & $(2N+1)\left\lceil \frac{2N+1}{n_d} \right\rceil$ REDC\\
d-\fqt    & $(2D_q+1)^2$ & $(2D_q+1)\left\lceil\frac{2D_q+1}{n_q}\right\rceil$ REDQ$_{2d+1}$ \\
\hline
\end{tabular}
\caption{Modular polynomial multiplication complexities.}\label{tab:mulpol}
\end{table}

For instance, with $p=3$, $N=500$, if we choose a double floating point
representation and a degree $4$ \dqt, the fully
tabulated \fqt boils down
to $10^5$ multiplications and additions and $4.10^3$ divisions. 
For the same parameters,
the classical
polynomial multiplication algorithm requires 
$10^6$ multiplications and additions and only 
$10^3$ remaindering. This is roughly $9$
times more operations as shown on figure~\ref{fig:dqt}. 
\begin{figure}[ht]
\hspace{-18pt}\includegraphics[width=\columnwidth*11/10]{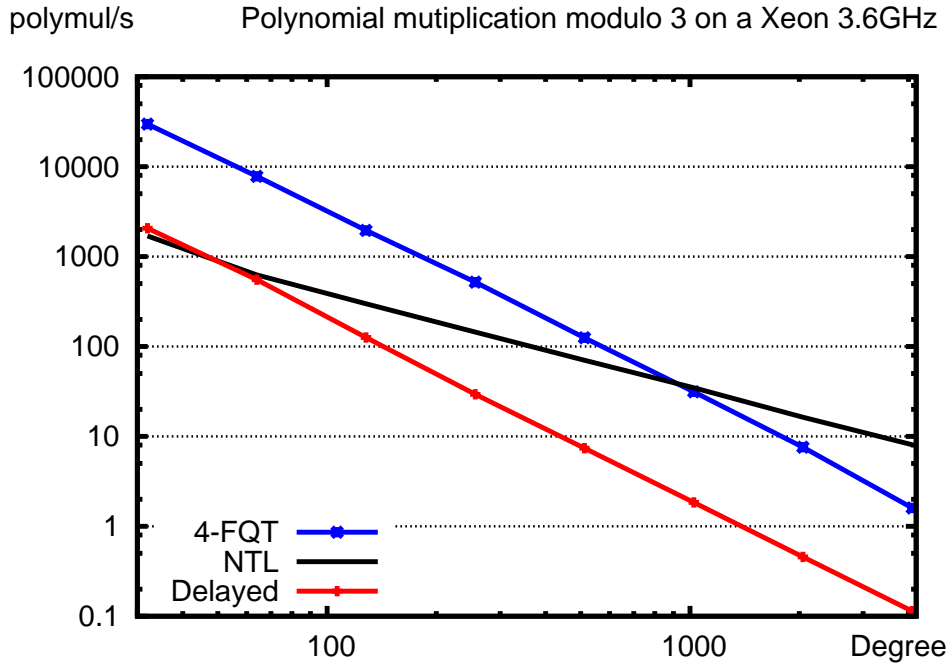}
\caption{Polynomial multiplications modulo 3 per second on a Xeon 
3.6 GHz}\label{fig:dqt}
\end{figure}

Even by switching to a larger mantissa, say e.g. 128 bits, so that the
\dqt multiplications are roughly 4 times costlier than double floating
point operations, this can still be useful:
take $p=1009$ and choose $d=3$, still gives around 
$10^5$ multiplications and additions over 128
bits and $4.10^3$ divisions. This makes $8$ times less operations. 
This should therefore still be faster than the
delayed over $32$ bits.

On figure \ref{fig:dqt},
we compare also our two implementations with that of NTL
\cite{Shoup:2007:NTL}. We see that
the \fqt is faster than NTL as long as better algorithms are not
used. Indeed the change of slope in NTL's curve reflects the use of
Karatsuba's algorithm for polynomial multiplication.
One should note that NTL also proposes a very optimized modulo 2
implementation which is an order of magnitude faster than our
implementation on small primes. There is therefore room for more
improvements on small fields.
Our strategy is anyway very useful for small degrees and small primes. 
Furthermore, we have not implemented the \fqt as the base case 
of faster recursive algorithms such as Karatsuba, Toom-Cook, etc.
The figure shows that these recursive algorithms together with the
\fqt could be the fastest.

In particular, the \fqt already
improves the speed of small finite field extension's arithmetic as
shown next.

\section{Application to small finite\\ field
  extensions}\label{sec:fflas}
The isomorphism between finite fields of equal sizes gives us a canonical
representation: 
any finite
field extension is viewed as the set of polynomials modulo a prime $p$
and modulo an
irreducible polynomial ${\cal P}$ of degree $k$. 
Clearly we can thus convert any finite field element to its $q$-adic
expansion ; perform the \fqt between two elements and then reduce the
obtained polynomial modulo ${\cal P}$.
Furthermore, it is possible to use floating point routines to perform
exact linear algebra as demonstrated in \cite{jgd:2008:toms}. 

Our
strategy here, see algorithm \ref{alg:FGDP}, is thus to convert
vectors over $\GF[p^k]$ to $q$-adic floating point, to call a fast
numerical linear algebra routine (BLAS) and then to convert the
floating point result back to the usual field representation.
In this paper we propose to improve all the conversion steps 
of \cite[algorithm 4.1]{jgd:2002:fflas}
 and thus approach the performance of the prime
field wrapping also for small extension fields:
\begin{enumerate}
\item Replace the Horner evaluation of the polynomials, to form the
  $q$-adic expansion, by a single table lookup, recovering directly the
  floating point representation.
\item Replace the radix conversion and the costly modular reductions
  of each polynomial coefficient, by a single REDQ operation.
\item Replace the polynomial division by two table lookups and
a single field operation.
\end{enumerate}
Indeed, 
suppose the internal representation of the extension field is already by
discrete logarithms and uses conversion tables from polynomial to index
representations. See e.g. \cite{jgd:2004:dotprod} for more details.
Then we choose a time-memory trade-off for the REDQ operation of the
same order of magnitude, that is to say $p^k$.
The overall memory required by these new tables only
doubles and the REDQ requires only $2$ accesses. 
Moreover, in the small extension, the polynomial multiplication must
also be reduced by an irreducible polynomial, ${\cal P}$. 
We show next that this reduction can be precomputed
in the REDQ table lookup and is therefore almost free.

Moreover, many things can be factorized if the field representation is
by discrete logarithms. Indeed, the element are represented by their discrete
logarithm with respect to a generator of the field, instead of by
polynomials. In this case
there are already some table accesses
for many arithmetic operations, see
e.g. \cite[\S 2.4]{jgd:2004:dotprod} for more details.

More precisely, we here propose algorithm~\ref{alg:FGDP} for linear
algebra over extension fields: line 1 is
the table look-up of floating point values associated to elements of
the field ; line 2 is the numerical computation ; line 3 to 7 is the
first part of the REDQ reduction ; line 8 and 9 are a time-memory
trade-off with two table accesses for the corrections of REDQ, combined
with a conversion from polynomials to discrete logarithm
representation ; the last line 10 combines the latter two results, inside
the field.
\begin{algorithm}[ht]
\caption{Fast Dot product over Galois fields via \fqt and \fqt inverse}
\label{alg:FGDP}
\begin{algorithmic}[1]
\REQUIRE a field $\GF[p^k]$ with elements represented as exponents of
a generator of the field.
\REQUIRE Two vectors $v_1$ and $v_2$ of elements of $\GF[p^k]$.
\REQUIRE a sufficiently large integer $q$.
\ENSURE $R \in \GF[p^k]$, with $R = v_1^T. v_2$.
\vspace{5pt}\newline{\underline{Tabulated $q-$adic conversion}}\vspace{2pt}
\newline\COMMENT{Use conversion tables from exponent to floating point evaluation}
\STATE Set $\widetilde{v_1}$ and $\widetilde{v_2}$ to the floating point vectors
of the evaluations at $q$ of the elements of $v_1$ and $v_2$.
\vspace{5pt}\newline{\underline{The floating point computation}}\vspace{2pt}
\STATE Compute $\tilde{r} = \widetilde{v_1}^T \widetilde{v_2}$;
\vspace{5pt}\newline{\underline{Computing a radix decomposition}}\vspace{2pt}
\STATE $r = \lfloor \tilde{r} \rfloor$;$~~~~$\COMMENT{$r=\tilde{r}$
  but we might need a conversion to an integral type}
\STATE $rop = \left\lfloor \frac{\tilde{r}}{p} \right\rfloor$;
\FOR{$i=0$ to $2k-2$}
\STATE $u_i = \left\lfloor \frac{r}{q^i} \right\rfloor -  p \left\lfloor \frac{rop}{q^i} \right\rfloor$;
\ENDFOR
\vspace{5pt}\newline{\underline{Tabulated radix conversion to exponents of the generator}}\vspace{2pt}
\newline\COMMENT{$\mu_i$ is such that $\mu_i = \widetilde{\mu_i} \mod p$ for
$\tilde{r} = \sum_{i=0}^{2k-2} \widetilde{\mu_i} q^i$} 
\STATE Set $L = representation( \sum_{i=0}^{k-2} \mu_i X^i )$.
\STATE Set $H = representation( X^{k-1} \times \sum_{i=k-1}^{2k-2} \mu_i X^{i-k+1} )$.
\vspace{5pt}\newline{\underline{Reduction in the field}}\vspace{2pt}
\STATE Return $R = H + L \in \GF[p^k]$;
\end{algorithmic}
\end{algorithm}

A variant of REDQ is used in algorithm~\ref{alg:FGDP},
but $u_i$ still satisfies $u_i =
\sum_{j=i}^{2k-2} \mu_j q^{j-i} \mod p$ as shown in
theorem~\ref{thm:REDQ}. Therefore the representations of 
$\sum \mu_i X^j$ in the field can be precomputed and stored in two tables
where the indexing will be made by $(u_0,\ldots,u_{k-1})$ and
$(u_{k-1},\ldots,u_{2k-2})$ and not by the $\mu_i$'s as shown next.
%
\begin{theorem}\label{thm:fgdp}
Algorithm~\ref{alg:FGDP} is correct.
\end{theorem}
\begin{proof}
There remains to prove that it is possible to compute $L$ and $H$ from
the $u_i$. From the equality above, we see that $\mu_{2k-2}=u_{2k-2}$ and 
$\mu_i = u_i-qu_{i+1} \mod p$, for $i=0..(2k-3)$. 
Therefore a precomputed table of $p^k$ entries, indexed
by $(u_0,\ldots,u_{k-1})$, can provide the representation of 
$$L=\sum_{i=0}^{k-2} (u_{i}-qu_{i+1} \mod p) X^i.$$
Another table with $p^k$ entries, indexed by
$(u_{k-1},\ldots,u_{2k-2})$, can provide the representation of
$$H=u_{2k-2}X^{2k-2}+\sum_{i=k-1}^{2k-3} (u_{i}-q u_{i+1} \mod p) X^{i}.$$

Finally $R = X^{k-1} \times \sum_{i={k-1}}^{2k-2} \mu_i X^{i-k+1}  +
\sum_{i=0}^{k-2} \mu_i X^i $ needs to be reduced modulo the
irreducible
polynomial used to build the field. But, if we are given the
representations of $H$ and $L$ in the field, $R$ is then equal
to their sum inside the field, directly using the internal
representations.
\end{proof}

Table~\ref{tab:comp} recalls the respective complexities of conversion
phase in the two
presented algorithms.
\begin{table}[ht]
\begin{center}
  \begin{tabular}[ht]{|l||c|c|c|}
\hline
   & Alg.~\ref{alg:dqt}& Alg.~\ref{alg:FGDP} &Alg.~\ref{alg:FGDP} \\
Memory 	& $3p^k$ & $6p^k$ &$4p^k+2^{k \lceil \log_2 p\rceil+1}$\\
\hline
\hline
Shift 	& $4k-2$ & $4k-2$ & $4k-2$\\
Add 	& $4k-4$ & $0$   & $2k-1$\\
Axpy 	& $0$  & $4k-3$ & $2k-1$ \\
Div 	& $2k-1$ & $0$  & $0$\\
Table 	& $0$  & $3$ & $3$\\
Red	& $\geq 5k$  & $4$  & $4$\\
\hline
  \end{tabular}
\caption{Complexity of the back and forth conversion between
  extension field and floating point numbers}\label{tab:comp}
\end{center}
\end{table}

\begin{figure*}[hbt]\center
\includegraphics[width=\textwidth*10/11]{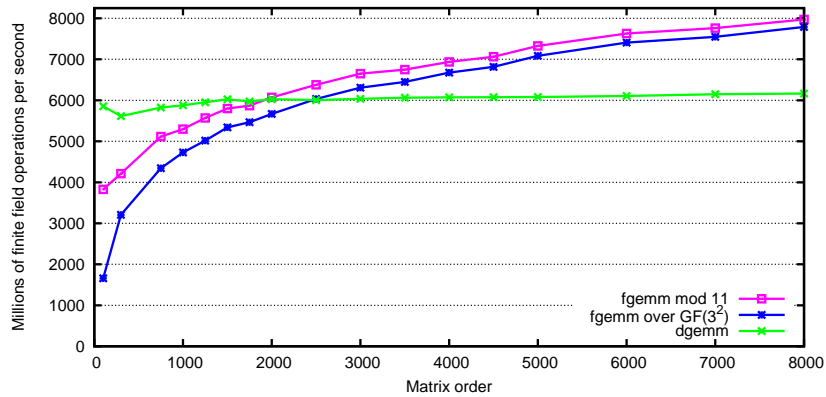}
\caption{Speed of finite field Winograd matrix multiplication on a XEON, 3.6 GHz}\label{fig:gfqgemm}
\end{figure*}

Figure~\ref{fig:conv} shows only the speed of the conversion
after the floating point operations. The log scales prove 
that for $q$ ranging from $2^1$ to
$2^{26}$ (on a 32 bit Xeon) our new implementation is two to three
times faster than the previous one.

\begin{figure}[ht]
\includegraphics[width=\columnwidth]{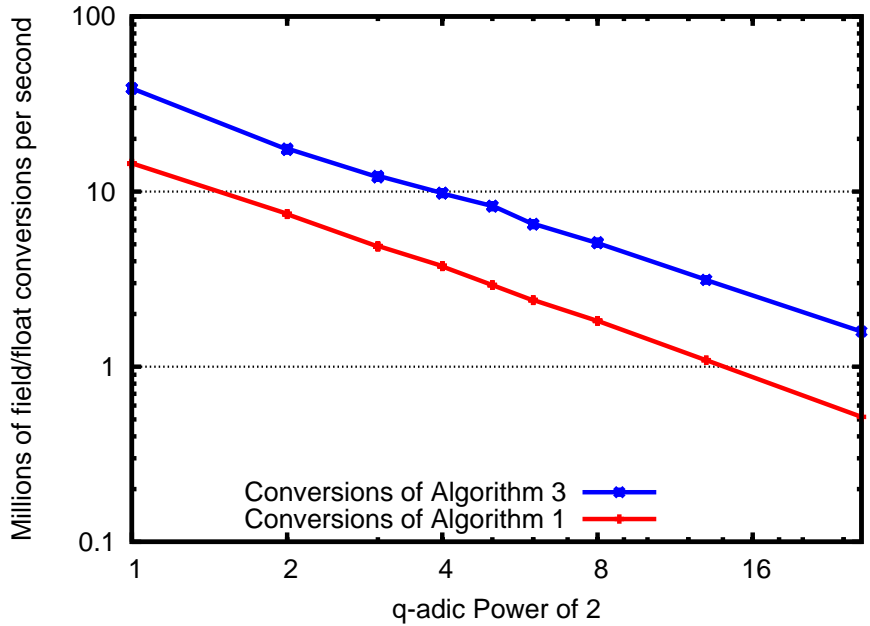}
\caption{Small extension field conversion speed on a Xeon 3.6GHz}\label{fig:conv}
\end{figure}


Furthermore, these improvements 
e.g. allow the extension field routines to reach the
speed of 7800 millions of $\GF[9]$ operations per second (on a XEON,
3.6 GHz, using Goto BLAS-1.09 \verb!dgemm! as the numerical routine
\cite{2002:gotoblas} and FFLAS \verb!fgemm! for the
fast prime field matrix multiplication \cite{jgd:2008:toms}) as shown
on figure~\ref{fig:gfqgemm}. The FFLAS routines are
available within the LinBox 1.1.4 library \cite{LinBox:2007} and the
\fqt is in implemented in the \verb!givgfqext.h! file of the 
Givaro 3.2.9 library \cite{Givaro:2007}.

With these new implementations, the obtained speed-up shown in figure
\ref{fig:gfqgemm}
represents a reduction from the 15 percent overhead of the
previous implementation to less than 4 percent now, when compared to
$\GF[11]$.

\section{Conclusion}
We have proposed a new algorithm for simultaneous reduction of several
residues stored in a single machine word.
For this algorithm we also give a time-memory trade-off implementation
enabling very fast running time if enough memory is available.

We have shown very effective applications of this trick for both
modular polynomial multiplication, and extension fields conversion to
floating point. The latter allows efficient linear algebra routines
over small extension fields but also linear algebra over small prime fields
as shown in \cite{jgd:2008:cmm}.

Further improvements include comparison of running times between
choices for $q$. Indeed our experiments were made with $q$ a
power of two and large table lookup. With $q$ a multiple of $p$ the
table lookup is not needed but divisions by $q^i$ will be more
expensive.

It would also be interesting to see how does the trick extend in
practice to larger precision implementations: on the one hand the basic
arithmetic slows down, but on the other hand the trick enables a more
compact packing of elements (e.g. if an odd number of field elements can
be stored inside two machine words, etc.).


\end{document}